\documentclass[journal,oneside,final,twocolumn,letterpaper]{IEEEtran}

\usepackage{epsf}
\usepackage{amssymb}
\usepackage{graphicx}
\usepackage{amsmath}
\usepackage[english]{babel}
\usepackage{mathrsfs}

\newtheorem{thm}{Theorem}
\newtheorem{lemma}[thm]{Lemma}
\newtheorem{definition}{Definition}

\newtheorem{algorithm}[definition]{Algorithm}

\newtheorem{problem}{Problem}
\newtheorem{simulation}{Simulation Setup}

\begin{document}

\title {Decentralized Sequential Change Detection Using Physical Layer Fusion}
\author {Leena Zacharias and Rajesh~Sundaresan
\thanks{Leena Zacharias is with Beceem Communications Pvt. Ltd., Bangalore, India, and Rajesh Sundaresan is with the Department of Electrical Communication
Engineering, Indian Institute of Science, Bangalore, India.}
\thanks{This work was supported by the Defence Research \& Development
Organisation (DRDO), Ministry of Defence, Government of India
under a research grant on wireless sensor networks (DRDO 571, IISc).}}

\maketitle

\begin{abstract}

The problem of decentralized sequential detection with conditionally
independent observations is studied. The sensors form a star
topology with a central node called fusion center as the hub. The
sensors make noisy observations of a parameter that changes from an
initial state to a final state at a random time where the random
change time has a geometric distribution. The sensors amplify and
forward the observations over a wireless Gaussian multiple access
channel and operate under either a power constraint or an energy
constraint. The optimal transmission strategy at each stage is shown
to be the one that maximizes a certain Ali-Silvey distance between
the distributions for the hypotheses before and after the change.
Simulations demonstrate that the proposed analog technique has lower
detection delays when compared with existing schemes. Simulations
further demonstrate that the energy-constrained formulation enables
better use of the total available energy than the power-constrained
formulation in the change detection problem.

\end{abstract}


\section{Introduction}
\label{sec:intro}

Consider the use of a wireless sensor network for detection of a disruption or a change in environment. The change is required to be detected with minimum delay subject to a false alarm constraint. The standard medium access control and physical layer design for such a network (e.g., IEEE 802.15.4 standard) is one where sensors quantize their observations and send them to a fusion center via random access over a wireless Gaussian multiple-access channel (GMAC). The transmitted data are typically quantized individual log-likelihood ratios (LLR) of the hypotheses representing the environment before and after the change. The fusion center collects each sensor's LLR and adds them to get a fused statistic, if observations at sensors are independent conditioned on the state of the environment; this would be the case when the observation noises are additive and independent from sensor to sensor\footnote{As we will see later, conditional independence notwithstanding, sensor observations are correlated.}. Such a design has a few drawbacks.
\begin{enumerate}
\item It does not exploit the spatial correlation in observations across sensors.
\item It does not exploit the superposition available on the GMAC.
\item It employs an ad hoc separation between quantization or compression on one hand, and transmission across the channel on the other; the latter requires adequate coding for noiseless reception and correct further processing at the fusion center.
\item It requires sufficient time slots for sensors to resolve all channel contentions\footnote{Alternatively, a time-division multiplexing protocol needs as many slots as there are sensors, and does not scale with the number of sensors.}.
\end{enumerate}

Our goal in this paper is to detect change in environment in a manner that addresses the aforementioned drawbacks. Specifically, we consider a ``star'' topology of sensors. Sensors make an affine transformation of the observed data and transmit the output in an analog fashion over the GMAC. Given that observations at sensors at any instant are spatially correlated, only the sum of the LLRs is relevant to the decision maker, i.e., it is a sufficient statistic to decide on the change. By making the sensors simultaneously transmit an affine function of their LLRs in an analog fashion, and via distributed transmit beamforming, we exploit the spatial correlation in sensor data and the superposition available on the GMAC -- the channel computes the required sum. Moreover, the analog data is in loose terms {\em matched} to the channel and does not require explicit channel coding. Finally, the sum is available at the fusion center in a single transmit duration unlike the situation in the random access case.

The biggest challenge in our proposed technique is the practicality of distributed transmit beamforming. The transmitters' clocks should be synchronized to some extent, so that carrier, phase, and symbol ticks align. A technique similar to the master-slave architecture proposed by Mudumbai, Barriac \& Madhow \cite{200705TWC_MudBarMad} can be used to achieve this synchronization. The scheme exploits channel reciprocity in a time-division duplex (TDD) system.

\subsubsection{Organization and preview of main results}
In Section \ref{sec:power}, we formulate and solve a change detection problem under a power-constrained setting\footnote{Sensors are usually powered by batteries with a fixed energy. The power-constrained model arises when this energy is evenly split over the desired life time of the sensor (in samples). An energy-constrained model arises when there is flexibility in how this energy is expended from sample to sample (subject to, of course, constraints imposed by the power amplifier).}. We arrive at a Markov decision problem framework and show that parameters of the affine transformation should minimize the variance of the combined observation and GMAC noises, which turns out to be a non-convex optimization problem. We then provide an explicit algorithm to compute the optimal control parameters. Section \ref{sec:energy}  considers an energy-constrained setting. Section \ref{sec:performance} compares the simulation performance of our scheme with a previously known scheme. It also compares the energy-constrained formulation of Section \ref{sec:energy} with the power-constrained formulation of Section \ref{sec:power}.
Appendix \ref{app:AliSilvey} contains a new characterization of optimal control: maximize a certain Ali-Silvey distance \cite{1966xxJRSS_AliSil} between the distributions of the fusion center's observation before and after the change. This is used to arrive at the minimum variance criterion of Section \ref{sec:power}.

\subsubsection{Prior work}
Change detection problems were solved in a centralized setting by Page \cite{195406BIO_Pag}, Lorden \cite{197112AMS_Lor}, and Shiryayev \cite{Shi-OSR78}. Shiryayev considered a Bayesian setting which is of relevance to our work.  Veeravalli \cite{200105TIT_Vee} solved the decentralized version of this problem with parallel error-free bit pipes of limited capacity from the sensors to the fusion center and identified the optimal stopping policy and quantizer structure. These results are analogous to those for hypothesis testing and sequential hypothesis testing (Tsitsiklis \cite{1990xxDD_Tsi}, Veeravalli et al. \cite{199303TIT_VeeBasPoo}). Prasanthi \cite{200606The_PraKum} considered access and decision delays in sequential detection over a random access channel, as it would be practically implemented using, for example, the IEEE 802.15.4 wireless personal area network standard. (See also \cite{200609SECON_PraKum}). Our work differs from those of Prasanthi and Veeravalli because we propose an analog transmission strategy.

Analog transmissions are optimal for transmission of a single Gaussian source over a Gaussian channel (Berger \cite[p.100]{Ber-RDT71}) and a bivariate Gaussian
source over a GMAC for a certain range of signal-to-noise ratios (SNR) (Lapidoth and Tinguely \cite{200607ISIT_LapTin}), when a running estimate is required.
Analog transmission via waveform design was considered by Mergen and Tong \cite{200602TSP_MerTon}. They used ``type-based'' multiple access to estimate a parameter over a GMAC. Their scheme, as does ours, exploits the superposition available in the GMAC. (See also \cite{200703TSP_MerNawTon}, \cite{200705TSP_LiuSay}, \cite{200707TSP_LiuGamSay}, \cite{200310Allerton_GasVet}, \cite{200305TIT_GasRimVet}, and \cite{200504JSAC_GasVet} for analog transmission in other settings). Ertin and Potter \cite{199802AFLS_ErtPot} considered generalized cost functions which is mathematically analogous to our
energy-constrained formulation.

\section{Physical Layer Fusion Framework}
\label{sec:power}

\subsection{Mathematical Formulation}

$X \sim \mathcal{N}(\theta,\sigma^2)$ indicates that $X$ is a Gaussian
random variable with mean $\theta$ and variance $\sigma^2$.

(1) The state of nature is described by $\left\{\theta_k : k \in
\mathbb{Z}_+\right\}$, a two-state discrete-time Markov chain taking
values in $\{ m_0, m_1 \}$, with transition probabilities as described in Fig. \ref{fig:setup}(a)-(b). The quantities $m_0$ and $m_1$ denote,
for example, the mean level of the observations before and after the
disruption. The initial distribution for this Markov chain is
obtained from $\Pr \{ \theta_0 = m_1 \} = \nu$.
The change time $\Gamma$ is $\mathbb{Z}_+$-valued, and given the event $\{ \Gamma > 0 \}$,
$\Gamma$ has the geometric distribution.

\begin{figure}
\centering
\includegraphics[height=4.5in,width=3.49in]{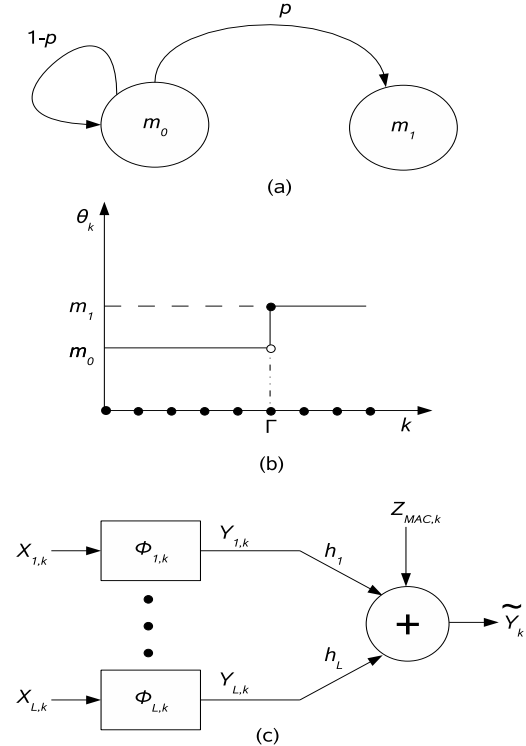}
\caption{Problem set-up.}
\label{fig:setup}
\end{figure}

(2) The network has $L$ sensors. At time $k$, sensor $S_l$ makes an
observation $X_{l,k} \sim \mathcal{N}(\theta_k,
\sigma_{\textsf{obs},l}^2)$, i.e.,
  $
     X_{l,k} = \theta_k + Z_{l,k},
  $
  where $Z_{l,k} \sim \mathcal{N}(0, \sigma_{\textsf{obs},l}^2), l = 1,\ldots,L$.

(3) The observations at each sensor are independent, conditioned on
$\theta_k$. Furthermore, the observations are independent
  from sensor to sensor, conditioned on $\theta_k$. Despite these conditional
independence assumptions, we remark that $X_{l,k}, l = 1, \ldots L$,
are correlated.

(4) Each sensor transmits $Y_{l,k} = \phi_{l,k} (X_{l,k})$; this
being a function only of the observation at sensor $l$, our setting
is a decentralized one. See Fig. \ref{fig:setup}(c). The function $\phi_{l,k}$ is affine:
\begin{equation}\label{eqn:affine}
\phi_{l,k}(x) = \alpha_{l,k}(x - c_{l,k}).
\end{equation}
Quantities $\alpha_k = (\alpha_{1,k},\ldots,\alpha_{L,k})$ and
$c_k = (c_{1,k},\ldots,c_{L,k})$ are parameters for optimal control.
Transmission is done by setting the amplitude of an
underlying unit-energy waveform to $Y_{l,k}$. All sensors use the
same underlying waveform. The motivations for the analog
amplify-and-forward transmissions in (\ref{eqn:affine}) are given in
Section \ref{sec:intro}: conditional independence of the
observations given the state, and the Gaussian observation noise. If
the latter does not hold, affine functions of LLRs instead of the
direct observations could be sent (\cite[Ch. 5]{200706The_Zac}).

(5) The GMAC output at the fusion center when projected onto the common waveform yields
\[
\widetilde{Y}_k = \sum_{l=1}^L h_lY_{l,k} + Z_{\textsf{MAC},k},
\]
where $Z_{\textsf{MAC},k} \sim \mathcal{N}(0,
\sigma_{\textsf{MAC}}^2)$ is independent and identically distributed
(iid) across $k$, and is independent of all other quantities. The
gain $h_l \in \mathbb{R}_+$ is the channel gain for the $l$th sensor
and is deterministic. See Fig. \ref{fig:setup}(c). We assume perfect knowledge of the channel
gains is available at the sensors and the fusion center. While this
is not the case in practice, channel knowledge can be gleaned in
time-division duplex (TDD) systems that possess channel reciprocity
(IEEE 802.15.4). See Mudumbai, Barriac \& Madhow
\cite{200705TWC_MudBarMad} for a suggested master-slave
architecture.
In a subsequent section, we study the effect of
imperfect knowledge of these gains.

(6) At the fusion center, form $\widehat{Y}_k$ as follows:
\begin{eqnarray}
\widehat{Y}_k  & = &  \frac{1}{\sum_{l=1}^Lh_l
\alpha_{l,k}}\left(\widetilde{Y}_k+
\sum_{l=1}^Lh_l\alpha_{l,k} c_{l,k}\right) \nonumber \\
\label{eqn:suffstat}& = & \theta_k + \widehat{Z}_{\textsf{MAC},k},
\end{eqnarray}
where $\widehat{Z}_{\textsf{MAC},k} \sim \mathcal{N} (0,\sigma_k^2 )$
and
\begin{equation}
\label{eqn:variance}\sigma_k^2 = \frac{\sum_{l=1}^L
(\sigma_{\textsf{obs},l}h_l\alpha_{l,k})^2+\sigma_{\textsf{MAC}}^2}
{\left(\sum_{l=1}^Lh_l\alpha_{l,k}\right)^2}.
\end{equation}
The quantity $\widehat{Y}_k$ in (\ref{eqn:suffstat}) is obtained
from $\widetilde{Y}_k$ using a bijective mapping; so no information
is lost. From (\ref{eqn:suffstat}), we also see that the distributed multi-sensor setting is
equivalent to a centralized setting where the fusion center makes a
direct (noisy) observation on $\theta_k$ with equivalent additive
observation noise of variance $\sigma_k^2$ as given in
(\ref{eqn:variance}). This is enabled by the affine nature of
$\phi_{l,k}$. The centralized problem with constant $\sigma_k^2$
was studied by Shiryayev \cite{Shi-OSR78} with
the aim of characterizing the stopping rule. The new aspect here is
the dependence of $\sigma^2_k$ on the control parameters.

(7) The fusion center chooses an action $a_{k-1} \in \mathbb{A}$ at time
$k-1$ from set $\mathbb{A}$ of actions (controls)
\[
  \mathbb{A} = \{ stop \} \cup \{ (continue, \alpha,
  c) : \alpha \in \mathbb{R}_+^L, c \in \mathbb{R}^L \}.
\]
If $a_{k-1} = stop$, the fusion center stops. If $a_{k-1} =
(continue, \alpha_k, c_k)$, the fusion center takes another sample
(the $k$th), and all sensors transmit $\phi_{l,k}(X_{l,k})$ with
parameters $(\alpha_k, c_k)$.

(8) As done by Veeravalli in \cite{199303TIT_VeeBasPoo}, we assume a
quasi-classical information structure, i.e., action $a_{k-1}$
depends on
\begin{equation}
\label{eqn:informationspace}
i_{k-1} = \left\{ a_0, \widehat{y}_1, a_1, \widehat{y}_2, \ldots,
a_{k-2}, \widehat{y}_{k-1} \right\}.
\end{equation}
Even though the sensors may have local memory of past observations,
our framework does not make use of this additional
information.\footnote{Veeravalli \cite[p.434]{199303TIT_VeeBasPoo}
discusses other information structures and why they may be difficult
to analyze.} The fusion center feeds back the action parameters
$a_{k-1}$ to the sensors. (We use the following notation: the
quantity $i_{k-1}$ in (\ref{eqn:informationspace}) is a realization
of the random variable $I_{k-1}$ and takes values in the set
$\mathbb{I}_{k-1}$. We set $\mathbb{I}_0 = \emptyset$).

(9) Average power constraint at sensor $l$ is
\[
\mathbb{E}\left[\alpha_{l,k}^2\left(X_{l,k}-c_{l,k}\right)^2|I_{k-1}\right]
 \leq  P_l,
\]
i.e.,
\begin{equation}
 \label{eqn:powerconstraint}\alpha_{l,k}^2\left[\sigma_{\textsf{obs},l}^2
+ \mathbb{E}\left[\left(\theta_k - c_{l,k}\right)^2 | I_{k-1}\right]
\right]  \leq  P_l, ~ l = 1,\ldots,L.
\end{equation}
The set of feasible controls, given $I_{k-1} = i_{k-1}$, is denoted by
\begin{eqnarray}
\lefteqn{ \mathbb{A}(i_{k-1}) = } \nonumber\\
\label{eqn:feasibleControls}
&  \{ stop \} \cup \{ (continue, \alpha, c) :
(\alpha, c, i_{k-1}) \mbox{ satisfies  (\ref{eqn:powerconstraint})} \}. &
\end{eqnarray}
In Section \ref{sec:energy}, we relax the constraint in (\ref{eqn:powerconstraint}) and impose an
expected total energy constraint.

(10) The fusion center policy $\pi$ is a sequence of proposed
(deterministic) actions $\pi = (\pi_{k-1}, k \geq 1)$, where
$\pi_{k-1}$ is a function $\pi_{k-1} : \mathbb{I}_{k-1} \rightarrow
\mathbb{A}$. In particular, $\pi_{k-1}(i_{k-1}) = a_{k-1} \in
\mathbb{A}(i_{k-1})$. Each policy $\pi$ induces a probability
measure. All expectations are with respect to this measure. The
dependence of the expectation operation on $\pi$ is understood and
suppressed.

(11) $\tau$ is the first instant when the fusion center decides to stop.

The problem we wish to solve is the following:

\vspace*{.1in}

\begin{problem} \label{prob:cdetdelay} \textit{(Change detection with delay
penalty)} Min\-imize over all admissible policies the expected
detection delay,
$
E_{\textsf{DD}} = \mathbb{E}\left[\left(\tau -
\Gamma\right)^+\right],
$
subject to an upper bound on the
probability of false alarm $P_{\textsf{FA}} \leq \delta$, where $x^+
 = \max(0,x)$, and $P_{\textsf{FA}} = \Pr \{ \tau < \Gamma \}
$.
\end{problem}

\vspace*{.1in}

The solution to Problem \ref{prob:cdetdelay} is obtained via a
solution to Problem \ref{prob:cdetbayes} (below) for a particular $\lambda >
0$ (Shiryayev \cite{Shi-OSR78}). The quantity $\lambda$ may be
interpreted as the cost of unit delay.

\vspace*{.1in}

\begin{problem} \label{prob:cdetbayes} \textit{(Change detection with a Bayes
cost)} Min\-imize over all admissible policies
\begin{eqnarray}
R(\lambda) & = & P_{\textsf{FA}}+\lambda
E_{\textsf{DD}}
 =  \Pr\{\Gamma > \tau\}+\lambda \mathbb{E}\left[\left(\tau-\Gamma\right)
^+\right] \nonumber \\
 \label{eqn:bayescost}
 & = &
 \mathbb{E} \left[ 1\{ \theta_{\tau} = m_0 \} + \sum_{k=0}^{\tau-1} \lambda 1\{ \theta_k = m_1 \} \right]
\end{eqnarray}
where $\lambda > 0$ and $\mathbb{E}$ is under the probability measure induced by the chosen policy.
\end{problem}

\vspace*{.1in}

The cost function is additive over time. The first term within the expectation in (\ref{eqn:bayescost}) is the terminal cost; the terms in the summation a running cost. At each stage the state $\theta_k$ evolves in a Markov fashion. The controller sees only a noisy version $\widehat{Y}_k$ of the state, but can control the observation noise variance $\sigma_k^2$ via $\alpha$ and $c$. It can also stop at any stage and pay a terminal cost. Any decision affects the future evolution of the cost process. Such problems are Markov decision problems (MDP) with partial observations. They can be analyzed by studying an equivalent complete observation MDP\footnote{See Shiryayev \cite{Shi-OSR78}, Veeravalli \cite{200105TIT_Vee} for results with stopping, Bertsekas \& Shreve \cite[Ch. 10]{1978SOC_BerShr} for discounted costs, and Bertsekas \cite[Ch. V]{Ber-DPOC95}.} with a reduced (posterior) state $
  \mu_k \stackrel{\Delta}{=} \mathbb{E} \left[ 1\{\theta_k = m_1\}
  \mid I_{k} \right] = \Pr \{ \Gamma \leq k \mid I_{k} \}.
$ The probability law for $\{ \mu_k : k \geq 0\}$ is given as
follows: $\mu_0 = \Pr \{ \Gamma \leq 0 \mid I_0 \} = \nu$, and the
law for $\mu_k$, under $a_k = (continue, \alpha_{k+1}, c_{k+1})$, is
(see Veeravalli \cite[eqn. (9)]{200105TIT_Vee})
\begin{eqnarray}
   \mu_{k+1}
      & = & \frac{\beta_kf_{m_1,\alpha_{k+1}} \left( \hat{Y}_
  {k+1} \right)}{\beta_kf_{m_1,\alpha_{k+1}} \left(\hat{Y}_{k+1}\right)+(1-\beta_k)
  f_{m_0,\alpha_{k+1}} \left(\hat{Y}_{k+1} \right)} \nonumber \\
  \label{eqn:recursion}  & \stackrel{\triangle}{=} & \frac{g \left( \hat{Y}_{k+1},
  \alpha_{k+1},\mu_k \right)}{h \left(\hat{Y}_{k+1},\alpha_{k+1},\mu_k \right)}
  \stackrel{\triangle}{=} \psi \left( \hat{Y}_{k+1}, \mu_k, \alpha_{k+1} \right),
\end{eqnarray}
where $ \beta_k \stackrel{\triangle}{=}\Pr\{\Gamma \leq k+1|I_k\} =
\mu_k+(1-\mu_k)p$, and $f_{m_i,\alpha_{k+1}}$ is the density of an
$\mathcal{N}(m_i, \sigma_{k+1}^2)$ random variable. The quantities
$h$ and $g$ are as in (\ref{eqn:recursion}); $h$ is the
density of $\hat{Y}_{k+1}$ given $(I_k, a_k)$, and $g$ is a scaled
density. The power constraint (\ref{eqn:powerconstraint}) when
written for time $k+1$ simplifies to
\begin{eqnarray}
  \alpha_{l,k+1}^2  \left[ \sigma^2_{\textsf{obs},l} + (m_0  -  c_{l,k+1})^2 (1 - \beta_k)\right.\nonumber\\
  \left. + (m_1 - c_{l,k+1})^2 \beta_k \right]
   \leq  P_l.
  \label{eqn:powerconstraint1}
\end{eqnarray}
The set of feasible controls in (\ref{eqn:feasibleControls}) depends
on $i_k$ only through $\mu_k$ and can be simplified to
\begin{eqnarray*}
  \mathbb{A}(\mu)  =  \{ stop \}  \cup  \{ (continue,  \alpha,  c)
  : \nonumber\\
  (\alpha,
  c,  \mu) \mbox{ satisfies } (\ref{eqn:powerconstraint1}) \},
\end{eqnarray*}
where $\mathbb{A}(\cdot)$ is re-used to denote the set of
feasible controls for the equivalent complete observation MDP.
Let
$
  \mathbb{A}'(\mu) = \{ (\alpha, c) : (continue, \alpha, c) \in \mathbb{A}(\mu) \}
$ denote the set of control  parameters when the action is to
continue. Now consider the objective function. Taking conditional
expectations with respect to the information process, (see Shiryayev
\cite[pp.195--196]{Shi-OSR78}), (\ref{eqn:bayescost}) reduces to
\begin{equation}
  \label{eqn:bayescost1}
  R(\lambda) = \mathbb{E} \left[ (1 - \mu_{\tau}) + \sum_{k=0}^{\tau-1} \lambda \mu_k \right].
\end{equation}
Minimization of (\ref{eqn:bayescost1}) is done via dynamic programming. Some additional remarks are in order.

{\it Remarks}: 1. The variance $\sigma_{k+1}^2$ depends on
$\alpha_{k+1}$ as shown in (\ref{eqn:variance}), and hence the
dependence on $\alpha_{k+1}$ in (\ref{eqn:recursion}). $\mu_{k+1}$
depends on $c_{k+1}$ only through $\alpha_{k+1}$ because of the
processing done in (\ref{eqn:suffstat}).

2. If the running cost is $\lambda$ instead of $\lambda 1\{ \theta_k
= m_1 \}$ in (\ref{eqn:bayescost}), every sample costs $\lambda$ units, not just those beyond
the change point that contribute to the delay. This is a minor
variation to Problem \ref{prob:cdetbayes} and has a similar solution.

3. Another variation is sequential hypothesis testing: set the transition probability
$p=0$, enhance the action $stop$ to $(stop, \hat{\theta})$, where
$\hat{\theta}$ is the decision (either $m_0$ or $m_1$), and set the
terminal cost to $1\{ \theta_{\tau} \neq \hat{\theta} \}$. The
running cost is a constant $\lambda$ for every sample.

\subsection{Optimal Policy}
\label{sec:policy}
As is usual with such problems, we first restrict the stopping time
$\tau$ to a finite horizon $T$. Using Bertsekas's result \cite[Ch.1,
Prop.3.1]{Ber-DPOC95}, the cost-to-go function recursions are
written as
\begin{eqnarray*}
J_T^T(\mu_T) & = & 1-\mu_T,\\
J_k^T(\mu_k) & = & \min\left\{1-\mu_k,\lambda \mu_k+A_k^T(\mu_k)\right\}, ~0 \leq k < T, \\
A_k^T(\mu) & = & \min_{(\alpha,c) \in \mathbb{A}'(\mu)} \mathbb{E}
\left[ J_{k+1}^T \left( \psi \left( \hat{Y}, \mu, \alpha \right)
\right) \right] \\
   &=&  \min_{(\alpha, c) \in \mathbb{A}'(\mu)} \! \int_{\mathbb{R}}
   \!
   {J_{k+1}^T\left(\frac {g\left(\hat{y},\alpha,\mu\right)}
   {h\left(\hat{y},\alpha,\mu\right)} \right) \! h\left(\hat{y},\alpha,\mu\right)d\hat{y}}.
\end{eqnarray*}

To solve Problem  \ref{prob:cdetbayes}, let $T\rightarrow
\infty$. From results in \cite{199303TIT_VeeBasPoo} and
\cite{200105TIT_Vee}, the limit in (\ref{eqn:bellman}) below exists, does not depend on $k$ (i.e., the policy is stationary),
and defines the infinite horizon cost-to-go function:
\begin{eqnarray}
J(\mu) = \lim_{T\rightarrow \infty}J_k^T(\mu)
\label{eqn:bellman} = \min \left\{1-\mu,\lambda\mu+A_J(\mu)\right\},
\end{eqnarray}
where
\begin{eqnarray}
\label{eqn:meancost} A_J(\mu) & = &
\min_{(\alpha,c) \in \mathbb{A}'(\mu)} \mathbb{E} \left[ J \left(
\psi \left( \hat{Y}, \mu, \alpha \right) \right) \right] .
\end{eqnarray}
The following lemma enables a characterization of the optimal
stopping policy.

\vspace*{.1in}

\begin{lemma}
\label{lemma:concavity} The functions $J_k^T(\mu)$ and $A_k^T(\mu)$
are non-neg\-ative and concave functions of $\mu$, for $\mu \in
[0,1]$. Moreover, $A_k^T(1) = J_k^T(1) = 0$. Similarly, the
functions $J(\mu)$ and $A_J(\mu)$ are non-negative and concave
functions of $\mu$, for $\mu \in [0,1]$, and $A_J(1) = J(1) = 0$.
\end{lemma}

\vspace*{.1in}

The proof is the same as that in Bertsekas \cite[p. 268]{Ber-DPOC95}
for sequential hypothesis testing. The concavity of $A_J(\mu)$ and
(\ref{eqn:bellman}) imply  the following theorem (Shiryayev \cite{Shi-OSR78}, Veeravalli
\cite{200105TIT_Vee}).

\vspace*{.1in}

\begin{thm}
\label{thm:stopping}
An optimal fusion center policy has stopping time $\tau$ given by
$
\tau = \inf \{k: \mu_k \geq \mu^*\},
$
where $\mu^*$ is the unique solution to
$
\lambda \mu +A_J(\mu) = 1-\mu.
$
\end{thm}

\vspace*{.1in}

To summarize, the optimal detection strategy at time $k$ is as
follows. Convert the received signal $\tilde{Y}_{k}$ into the
posterior probability of change $\mu_{k}$ using (\ref{eqn:suffstat})
and (\ref{eqn:recursion}). If $\mu_{k}$ exceeds a threshold, declare
that a change has occurred. Otherwise, make the sensors transmit
another sample using parameters $\alpha,c$ chosen optimally as
described in the next subsection.

\subsection{Parameters for Optimal Control}
\label{sec:controls}

We begin this section with an algorithm that calculates the optimal $\alpha$.

\vspace*{.1in}

\begin{algorithm}
\label{alg:controls} Let
\[ \sigma_{\textsf{obs},1}^2h_1\alpha_{\textsf{max},1} \leq \cdots
\leq \sigma_{\textsf{obs},L}^2h_L\alpha_{\textsf{max},L},\]
where the quantity
\[ \alpha_{\textsf{max},l} = \left(P_l/ \left(\sigma_{\textsf{obs},l}^2+(m_1-m_0)^2 \beta(1-\beta)\right)\right)^{1/2} \]
with $\beta = \mu+(1-\mu)p$.
\begin{itemize}
\item
\textbf{Step 1}: Find the unique $k \in \{1,\ldots,L-1\}$ that satisfies
\begin{eqnarray}
\sigma_{\textsf{obs},k}^2h_k\alpha_{\textsf{max},k} & \leq &
\frac{\sum_{l=1}^k(\sigma_{\textsf{obs},l}h_l\alpha_{\textsf{max},l})^2+\sigma_{\textsf{MAC}}^2}
{\sum_{l=1}^kh_l\alpha_{\textsf{max},l}} \nonumber \\
 \label{eqn:condition} & \leq &
\sigma_{\textsf{obs},k+1}^2h_{k+1}\alpha_{\textsf{max},k+1}
\end{eqnarray}
if it exists. Otherwise, set $k = L$.
\item
\textbf{Step 2}: Set the optimal $\alpha$ as follows.
\begin{eqnarray}
\lefteqn{ a^*  =  \sum_{l=1}^k h_l\alpha_{\textsf{max},l} +
\frac{\sum_{l=k+1}^L \sigma_{\textsf{obs},l}^{-2}}{\sum_{l=1}^k h_l \alpha_{\max, l}} } \nonumber \\
\label{eqn:astar} && ~\cdot \left( \sum_{l=1}^k(\sigma_{\textsf{obs},l}h_l\alpha_{\textsf{max},l})^2 + \sigma_{\textsf{MAC}}^2 \right),
\end{eqnarray}
\begin{eqnarray}
\alpha_m & = & \alpha_{\textsf{max},m}, ~ 1 \leq m \leq k, \nonumber \\
\alpha_m & = & \frac{1}{\sigma_{\textsf{obs},m}^2 h_m } \cdot
\frac{a^*-\sum_{l=1}^kh_l\alpha_{\textsf{max},l}}
{\sum_{l=k+1}^L \sigma_{\textsf{obs},l}^{-2}},
~ k < m. \nonumber \\
&&\label{eqn:alpha}
\end{eqnarray}
\end{itemize}
\end{algorithm}

\vspace*{.1in}

The optimal choice sets amplitudes of the $k$ sensors with the
$k$ least scaled observation noise variance ($\sigma_{\textsf{obs},l}^2 h_l \alpha_{\max,l}$) to $\alpha_{\textsf{max},l}$. The remaining
sensors' amplitudes are appropriately chosen smaller values.
Intuitively, sensors $l = k+1, \ldots, L$ have so good a channel
that scaling by $\alpha_{\textsf{max},l}$ for these sensors will
amplify the observation noise leading to a larger overall noise
variance. Note that when all channel gains, observation variances,
and power constraints are equal,
$\alpha_l = \alpha_ {\textsf{max}}$
for all sensors. This special case was earlier proved in
\cite{200706ISIT_ZacSun}.

\vspace*{.1in}

\begin{thm}\label{thm:controls}
The choice of
$
c_l =  m_1\beta+m_0(1-\beta), ~l = 1,\ldots,L,
$
and $\alpha$ according to Algorithm \ref{alg:controls} constitute the
optimal controls that minimize (\ref{eqn:meancost}).
\end{thm}

\vspace*{.1in}

\begin{proof} \textbf{Step 1}: We prove that the optimal control minimizes the variance
(\ref{eqn:variance}). Consider $\alpha$ and $\alpha'$ with resulting
variances $\sigma^2 < \sigma'^2$. From the second equality in
(\ref{eqn:suffstat}) we have
\begin{eqnarray}
\label{eqn:randprocessing1}\widehat{Y}(\alpha) = \theta + \sigma Z,
\end{eqnarray}
\begin{eqnarray}
\widehat{Y}(\alpha')  =  \theta + \sigma'Z'
\label{eqn:randprocessing2} ~~ \sim ~~ \theta+\sigma Z_1+(\sigma'^2-\sigma
^2)^{1/2}Z_2,
\end{eqnarray}
where $Z,~Z_1,~Z_2$ are iid $\mathcal{N}(0,1)$ with $Z' = Z_1+Z_2$. The time index $k$ is understood.

From (\ref{eqn:randprocessing1}) and (\ref{eqn:randprocessing2}),
$\widehat{Y}(\alpha')$ is a stochastically degraded
version of $\widehat{Y}(\alpha)$ and is equivalent to an additional
random processing on $\widehat{Y}(\alpha)$.
Theorem \ref{thm:AliSilvey} in Appendix \ref{app:AliSilvey} shows that
\[
1-\mathbb{E}_h\left[J\left(\frac{g\left(\widehat{Y}(\alpha),\alpha,\mu\right)}
{h\left(\widehat{Y}(\alpha),\alpha,\mu\right)}\right)\right]
\]
is an {\em Ali-Silvey distance} between two probability measures. In
$\mathbb{E}_h$ the dependence  of $h$ on $\alpha$ is understood and
suppressed. Ali-Silvey distances have a well-known monotonicity
property: data processing, whether deterministic or random, cannot
increase the dissimilarity measure between two distributions
(\cite{1966xxJRSS_AliSil}, \cite{1967xxSSMH_Csi}). This property
implies that
\[
\mathbb{E}_h \! \left[ \! J \! \left(\!
\frac{g\!\left(\!\widehat{Y}(\alpha),\alpha,\mu\!\right)}
{h\!\left(\!\widehat{Y}(\alpha),\alpha,\mu\!\right)} \! \right) \!
\right] \! \leq \! \mathbb{E}_h \! \left[ \! J \! \left(\!
\frac{g\!\left(\!\widehat{Y}(\alpha'),\alpha',\mu\! \right)}
{h\!\left(\!\widehat{Y}(\alpha'),\alpha',\mu\!\right)} \! \right)\!
\right].
\]

 It follows that minimization of the variance in
(\ref{eqn:variance}) is the criterion for getting the optimal
$\alpha$.

\textbf{Step 2}: We now identify the optimal $c$.
The minimization mentioned in the previous step
should be done subject
to the power constraint given in (\ref{eqn:powerconstraint}), which can be
rewritten as
\begin{equation}
\label{eqn:alphasquarebound} \alpha_{l,k}^2 \leq
P_l \cdot \left[\sigma_{\textsf{obs},l}^2 +
\mathbb{E}\left[\left(\theta_k - c_{l,k}\right)^2 | I_{k-1}\right]
\right]^{-1}.
\end{equation}
The constraint set is enlarged if the upper bound in
(\ref{eqn:alphasquarebound}) is higher.  We should therefore choose
the $c_{l,k}$ that minimizes $\mathbb{E}\left[\left(\theta_k -
c_{l,k}\right)^2 | I_{k-1}\right]$, i.e., $c_{l,k}$ is the minimum
mean squared error (MMSE) estimate of $\theta_k$ given $I_{k-1}$.
Clearly this is given by $ c_{l,k}  =
\mathbb{E}\left[\theta_k|I_{k-1}\right]
 =  m_1\beta_{k-1}+m_0(1-\beta_{k-1}),
$ and is independent of $l$. Moreover,
\begin{eqnarray*}
\mathbb{E}\left[\left(\theta_k - c_{l,k}\right)^2 | I_{k-1}\right] & = &
\textsf{Var}\left\{\theta_k|I_{k-1}\right\} \\
 & = &  (m_1-m_0)^2\beta_{k-1}(1-\beta_{k-1}),
\end{eqnarray*}
and (\ref{eqn:alphasquarebound}) can be written as $ \alpha_{l,k}
\leq \alpha_{\textsf{max},l,k}, $ where
\[
  \alpha_{\textsf{max},l,k} = \left( P_l/\left(\sigma_{\textsf{obs},l}^2
  +(m_1-m_0)^2\beta_{k-1}(1-\beta_{k-1})\right) \right)^{1/2}.
\]

\textbf{Step 3}:
Ignoring the time index $k$, the optimization problem to obtain
the best $\alpha$ is:

\vspace*{.1in}

\begin{problem}
\label{prob:optialpha}
Minimize
\[
\left(\sum_{l=1}^Lh_l\alpha_l\right)^{-2} \left[ \sum_{l=1}^L(\sigma_{\textsf{obs},l}h_l\alpha_l)^2+\sigma_{\textsf{MAC}}^2 \right],
\]
where $\alpha_l \in [0, \alpha_{\textsf{max},l}]$ for $l=1, \cdots, L$.
\end{problem}

\vspace*{.1in}

This is not a convex optimization problem. However, we can split it
into two simpler convex optimization problems to get
an explicit solution to Problem \ref{prob:optialpha}.

\vspace*{.1in}

\begin{lemma}\label{lemma:optimization}
Algorithm \ref{alg:controls} solves Problem \ref{prob:optialpha}.
\end{lemma}

\vspace*{.1in}

See Appendix \ref{app:optimization} for a proof. This concludes the
proof of Theorem \ref{thm:controls}.
\end{proof}

\vspace*{.1in}

Under the restriction of affine controls,
Theorem \ref{thm:controls} describes the
optimal choice. However, affine controls are not optimal in general.
This is demonstrated in Fig. \ref{fig:clipping} where a piece-wise
linear sigmoidal control outperforms the optimal affine control (see
\cite[Sec. 2.7]{200706The_Zac}).  It would be interesting to see if
there are ranges of $\sigma^2_{\textsf{obs},l}$ and
$\sigma^2_{\textsf{MAC}}$ where the affine control is indeed
optimal. We do not pursue this question in this work.

\begin{figure}
\centering
\includegraphics[height=2.6in,width=3.49in]{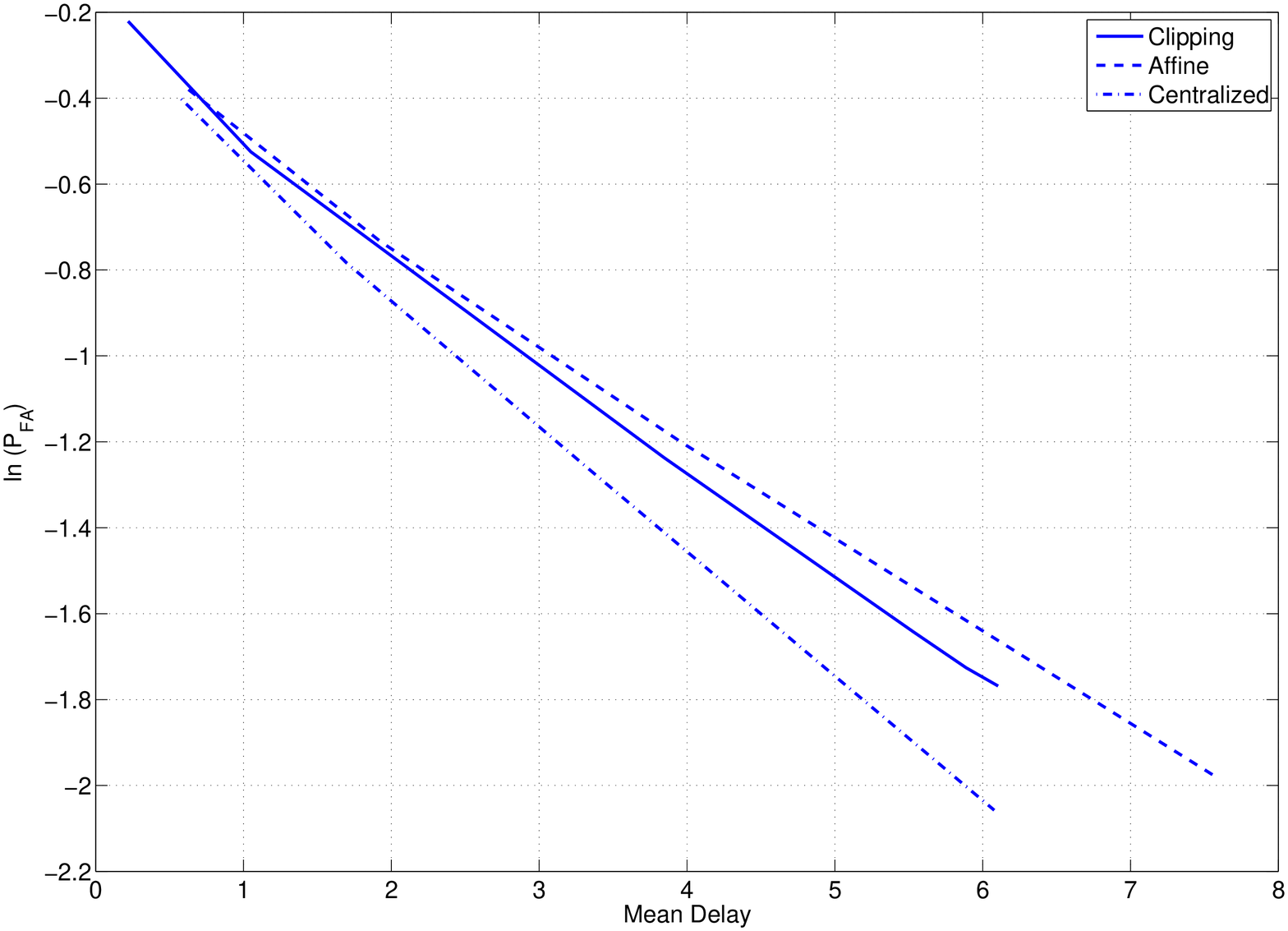}
\caption{Performance curves: 1) Clipped transmission via a sigmoidal
function 2) Affine transformation 3) Centralized, where all sensor
data is available without noise at the fusion center.}
\label{fig:clipping}
\end{figure}

We now make some remarks on the complexity of overall detection.
Theorem \ref{thm:controls} says that the parameters for optimal
control are obtained via a finite step procedure. Indeed, Algorithm
\ref{alg:controls} gives the output in time linear in the number of
sensors, and is therefore easy to execute. The threshold calculation
for a fixed set of parameters is a one time calculation and is
obtained via the so-called {\it value iteration procedure} which
yields an approximation. We now explore further simplifications with
reduced feedback information.

\subsection{A Simpler Suboptimal Policy}
\label{sec:suboptimal}

Let us now restrict the controls to be of the following form: the
decision to stop or continue, say $b_k$, depends on $I_k$, but the
parameters of the affine transformation at time $k+1$ can only
depend on $I_0$ and $b_k \in \{stop, continue\}$. $I_0$ denotes the
prior information before any observations are made and $b_k$ is the
decision of the fusion center at $k$. Note that this reduces the
amount of feedback to simply the binary random variable $b_k$.

The structure of the controls is similar to that of the optimal
policy of the previous section, but with
\[\beta_k =
\Pr\left\{\Gamma\leq k+1|I_0\right\} = 1-(1-\nu)(1-p)^{k+1}\]
so that
$(\alpha,c)$ depends on only $I_0$ and not on $I_k$. The stopping
policy is chosen as in Theorem \ref{thm:stopping}. As we see in
simulation results presented in Section \ref{sec:performance}, the
performance of this algorithm is close to optimal for the chosen
parameters, yet requires feedback of only one bit at each stage.

\section{Energy-Constrained Formulation}
\label{sec:energy}
The energy-constrained problem is stated as follows.

\vspace*{.1in}

\begin{problem} \label{prob:energy}
Minimize the expected detection delay, $E_{\textsf{DD}}$, subject to
an upper bound on the probability of false alarm, $P_{\textsf{FA}} \leq
\delta$, and an
upper bound on the expected energy spent,
\begin{equation}
\label{eqn:energyConstraint}
\mathbb{E}\left[\sum_{k=1}^{\tau}\mathbb{E}\left[\phi_{l,k}^2(X_{l,k})
|I_{k-1}\right]\right] \leq E_l, \quad  l=1,2,\ldots,L.
\end{equation}
\end{problem}

\vspace*{.1in}

Let $\lambda=(\lambda_1,\ldots,\lambda_L,\lambda_{L+1})$. As before,
to solve Problem \ref{prob:energy}, we set up  the Bayes cost
$R(\lambda)$ and minimize it over all admissible choices
of stopping policy and the parameters $\alpha_{l,k}$ and $c_{l,k}$ of the affine transformation
$\phi_{l,k}$. The Bayes cost can be
written as
\begin{eqnarray*}
R(\lambda) & = &  \mathbb{E}\Big[ (1 - \mu_{\tau}) + \lambda_{L+1}\sum_{k=0}
^{\tau-1} \mu_k  \\
&&  ~+~ \sum_{k=1}^{\tau}\sum_{l=1}^L\lambda_l
\mathbb{E}\left[\alpha_{l,k}^2(X_{l,k}-c_{l,k})^2|I_{k-1}\right]\Big].\nonumber
\end{eqnarray*}
A result analogous to Theorem \ref{thm:stopping} in Section \ref{sec:policy} holds, and the optimal
control at time $k+1$, given $I_k$, is such that $c_{k+1}$ is independent
of $l$, the sensor index. More precisely,
\begin{eqnarray}
  \label{eqn:minc-energyConstrained}
c_{k+1}  =  m_1\beta_k+m_0(1-\beta_k),\quad l = 1,\ldots,L,\nonumber
\end{eqnarray}
\begin{eqnarray}
\alpha_{k+1}\!\!\! & = & \!\!\! \arg \! \min_{\alpha \in
\mathbb{R}_+^L} \! \Bigg[ \! \sum_ {l=1}^L \!
\lambda_l\alpha_l^2\left(\sigma_{\textsf{obs},l}^2\!+\!
(m_1 \! - \! m_0)^2\beta_k(1 \! - \! \beta_k)\right)  \nonumber \\
 & &+\! \int_{\mathbb{R}} \! {J \!
\left(\!\frac{g\left(\hat{y},\alpha,\mu_k\right)}{h\left
(\hat{y},\alpha,\mu_k\right)}\!\right) \!
h\!\left(\hat{y},\alpha,\mu_k\right)d\hat{y}}\Bigg]\nonumber,
\end{eqnarray}
where  $J(\mu)  =  \min
\left\{1-\mu,\lambda_{L+1}\mu+A_J(\mu)\right\},$ is the infinite
horizon cost-to-go function with
\begin{eqnarray}
 \label{eqn:meancostenergy} A_J(\mu)\! & = & \! \min_{\alpha \in \mathbb{R}_+^L}
\Big[\sum_
{l=1}^L\lambda_l\alpha_l^2\left(\sigma_{\textsf{obs},l}^2+(m_1-m_0)^2\beta(1-\beta)\right) \nonumber \\
& & +
 \int_{\mathbb{R}} {J\left(\frac
{g\left(\hat{y},\alpha,\mu\right)}{h\left(\hat{y},\alpha,\mu
\right)}\right)h\left(\hat{y},\alpha,\mu\right)d\hat{y}}\Big]\nonumber.
\end{eqnarray}
A minimizing control $\alpha$ does exist as is shown in \cite[Sec. 3.1]{200706The_Zac}.

\section{Comparisons and practical considerations}
\label{sec:performance}

\subsection{Benefits from Exploiting Sensor Correlation}

Veeravalli \cite{200105TIT_Vee} addresses the structure of optimal
$D_l$-level quantizer at sensor $S_l,l=1,2,\ldots,L$. His model is
applicable to a system that allows $\log_2 D_l$ bits to be sent
error-free from sensor $S_l$ to the fusion center. For simplicity
let $D_l=D,l=1,2,\ldots,L$. In order to show the benefit of
exploiting correlation of observations when transmitting across the
GMAC, we do the following. The quantized bits from the sensors in
Veeravalli's scheme are transmitted using an optimal scheme designed
for independent data streams over a coherent GMAC. If all sensors operate at
the same transmission power, the $\textsf{SNR}$ required to support
such a transmission on the GMAC satisfies the sum rate constraint $
  L\log_2 {D} \leq (1/2) \log_2 {\left(1+L\cdot\textsf{SNR}\right)},
$
and thus
\begin{equation}
\label{eqn:GMACsnr}\textsf{SNR} \geq \frac{D^{2L}-1}{L}.
\end{equation}

For the simulations, we assume two sensors ($L=2$) with equal gains,
i.e., $h_l = 1$ for $l=1,2$. We also assume one-bit quantizers
($D=2$). From (\ref{eqn:GMACsnr}) we get $\textsf{SNR} \geq 7.5$.
Algorithms operate at $\textsf{SNR} = 7.5$ with $P_l = 7.5$ for
$l=1,2$ and $\sigma_{\textsf{MAC}}^2 = 1$. We now summarize the
other simulation assumptions which will be used unless stated
otherwise.

\vspace*{.1in}

\begin{simulation}\label{simulation}
Consider $L=2$ sensors with $\mathcal{N}(0,1)$ and
$\mathcal{N}(0.75,1)$ observations before and after the change,
respectively. The geometric parameter $p=0.05$ and the initial
probability of change $\nu=0$. We obtain $J(\mu)$ via value
iteration procedure until the difference between successive iterates
falls below $0.0001$ with $1000$ points on the $\mu$ axis. All
simulations assume $P_l = P$ and $\sigma_{\textsf{obs},l}^2 = 1$ for $l=1,2$.
\end{simulation}

\vspace*{.1in}

Fig. \ref{fig:veeravalli} shows that both our algorithms give lesser delays
than Veeravalli's algorithm that is naively overlaid on the GMAC. Furthermore, the suboptimal
policy of Section \ref{sec:suboptimal} degrades from that in Section
\ref{sec:policy} only for low false alarm probabilities.

\begin{figure}
\centering
\includegraphics[height=2.6in,width=3.49in]{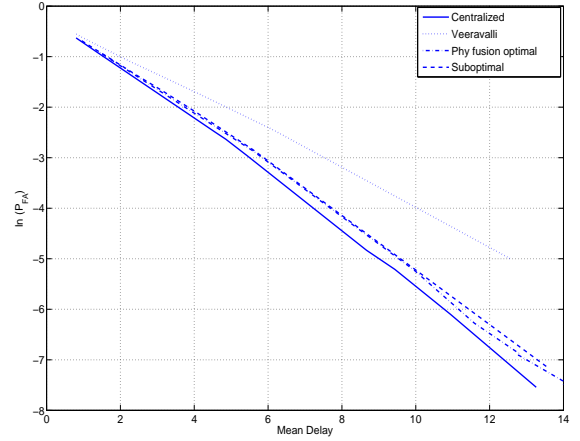}
\caption{Comparison of our algorithms with Veeravalli's scheme. The
``centralized'' performance curve is for the case when all sensor
data is available without noise at the fusion center.} \label{fig:veeravalli}
\end{figure}

In Veeravalli's algorithm,
$D-1$ thresholds ($\in \mathbb{R}^{D-1}$)
and a decision to stop or continue are fed back to each sensor. Our scheme requires
feedback of $\alpha_l \in \mathbb{R}_+, c_l \in \mathbb{R}$, and
the binary decision.
Even simpler is the strategy in Section \ref{sec:suboptimal}; only a
binary decision is fed back.

The network delay is independent of the number of sensors in both
our algorithms; the performance improves with increasing number of
sensors. Veeravalli's scheme on the other hand requires an
exponential growth in \textsf{SNR} (with $L$, as in
(\ref{eqn:GMACsnr})) to maintain the same delay versus
$P_{\textsf{FA}}$ performance. Our algorithms need a higher level of time and frequency synchronization of the transmitters for
beamforming. Section \ref{subsec:channelerrors} studies the
effect of lack of perfect channel knowledge. Transmit beamforming
can be achieved via uplink-downlink reciprocity in a static time-division
duplex (TDD) system (see \cite{200705TWC_MudBarMad} for an example mechanism).

\subsection{Performance Comparisons Under Different Channel and Observation SNRs}
We now portray performance under three different settings.
\begin{itemize}
  \item Fig. \ref{fig:phy_comp_chanSNR} shows performance for various channel SNRs $P/\sigma_{\textsf{MAC}}^2$; the other parameters remain as in Simulation Setup \ref{simulation}.
  \item Fig. \ref{fig:phy_comp_obsSNR} shows performance for various observation SNRs $(m_1-m_0)^2/\sigma_{\textsf{obs}}^2$ when the channel SNR $P/\sigma_{\textsf{MAC}}^2$ is fixed at 3 dB.
  \item Fig. \ref{fig:phy_asym} compares the symmetric and asymmetric channel
gain cases. The symmetric curve is obtained with $h_l = 1$ for $l = 1,2$, and the asymmetric one with $h_1=1$ and $h_2=0.75$. The weaker sensor is 2.5 dB lower than the stronger one.
\end{itemize}
The plots show graceful degradation with decreasing SNR with results along expected lines.

\begin{figure}
\centering
\includegraphics[height=2.65in,width=3.49in]{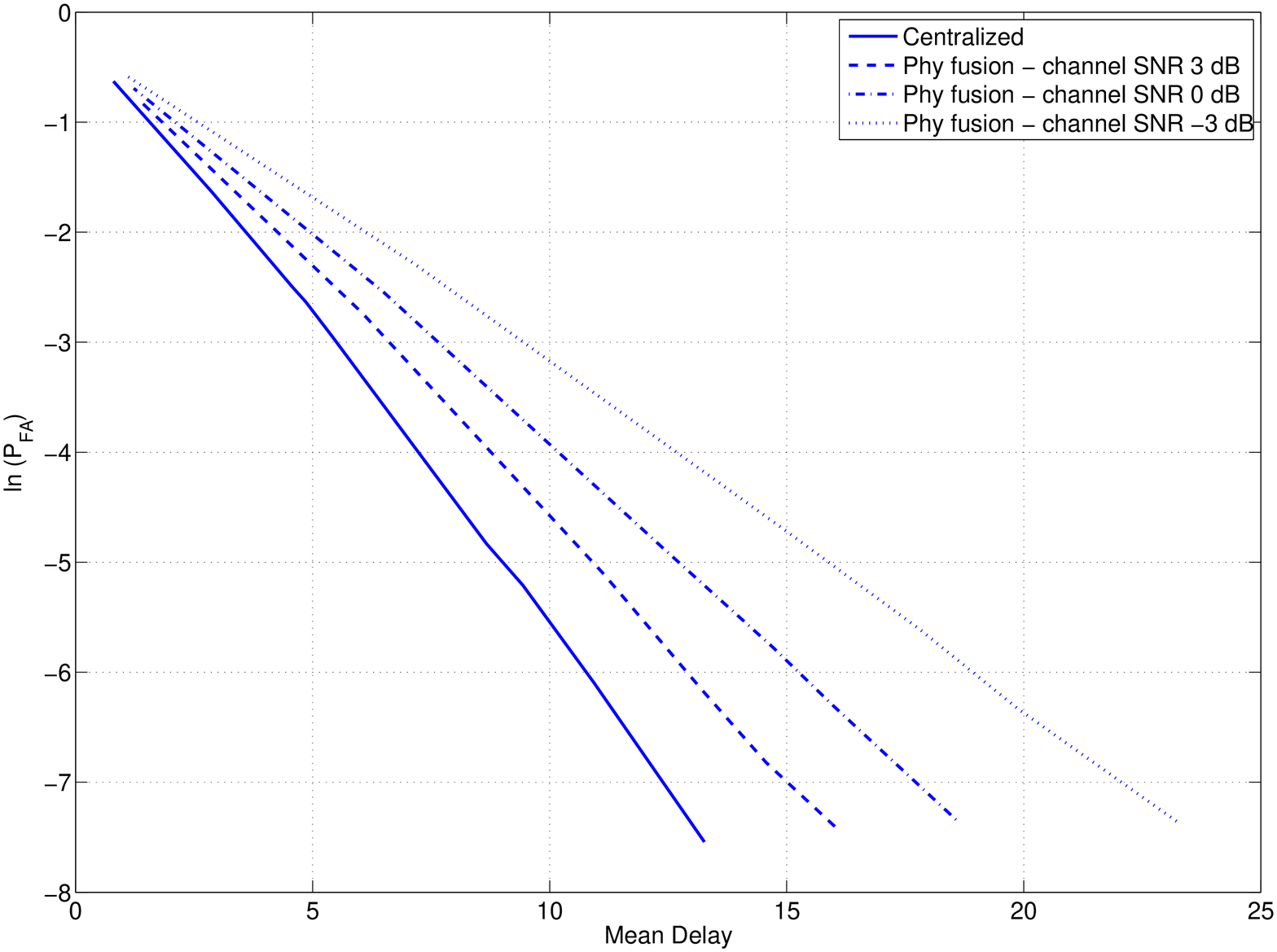}
\caption{Performance curves for channel SNR = $\infty, 3, 0, -3$
dB.} \label{fig:phy_comp_chanSNR}
\end{figure}

\begin{figure}
\centering
\includegraphics[height=2.65in,width=3.49in]{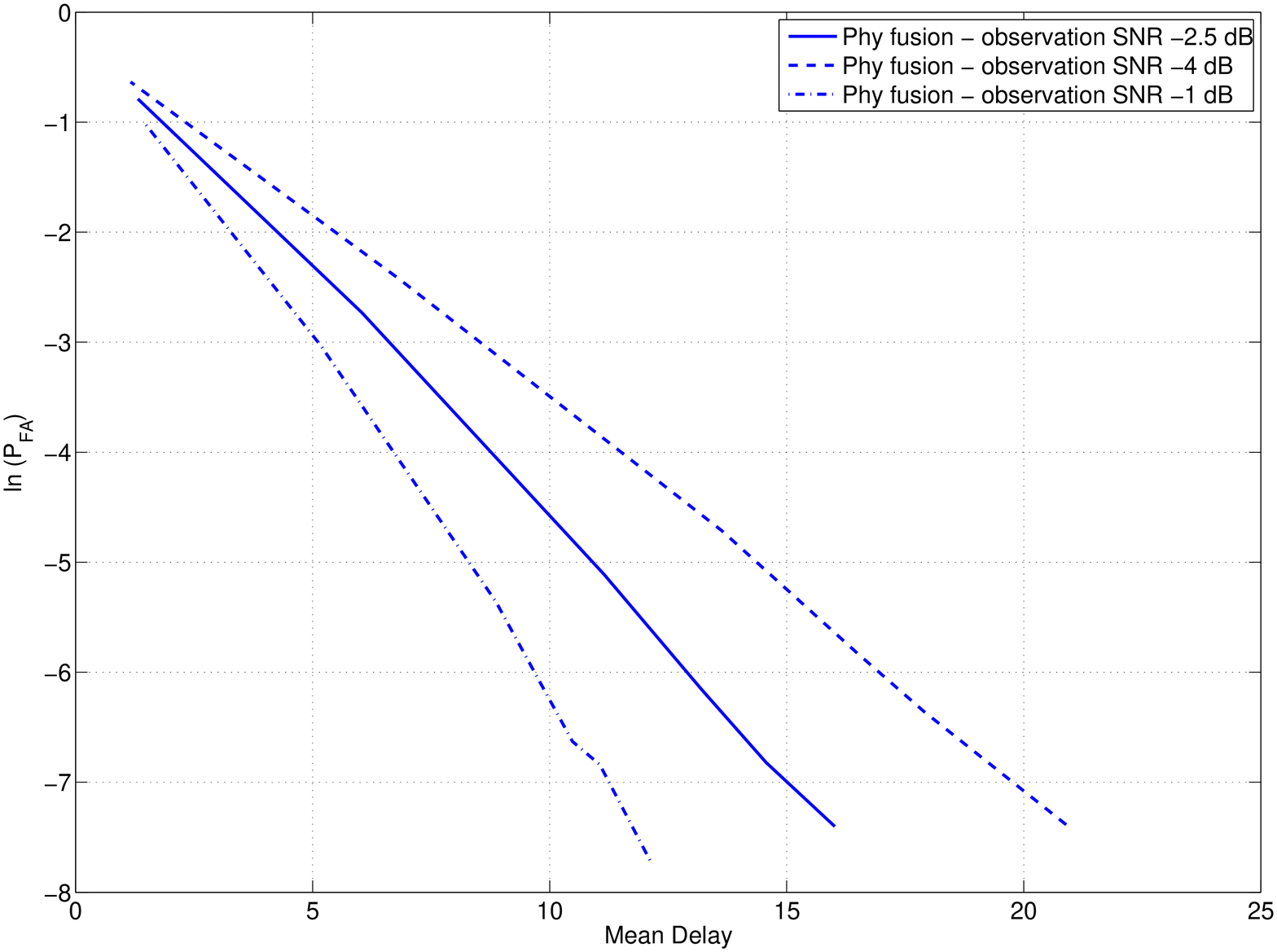}
\caption{Performance curves for observation SNR = $-1, -2.5, -4$
dB.} \label{fig:phy_comp_obsSNR}
\end{figure}

\begin{figure}
\centering
\includegraphics[height=2.65in,width=3.49in]{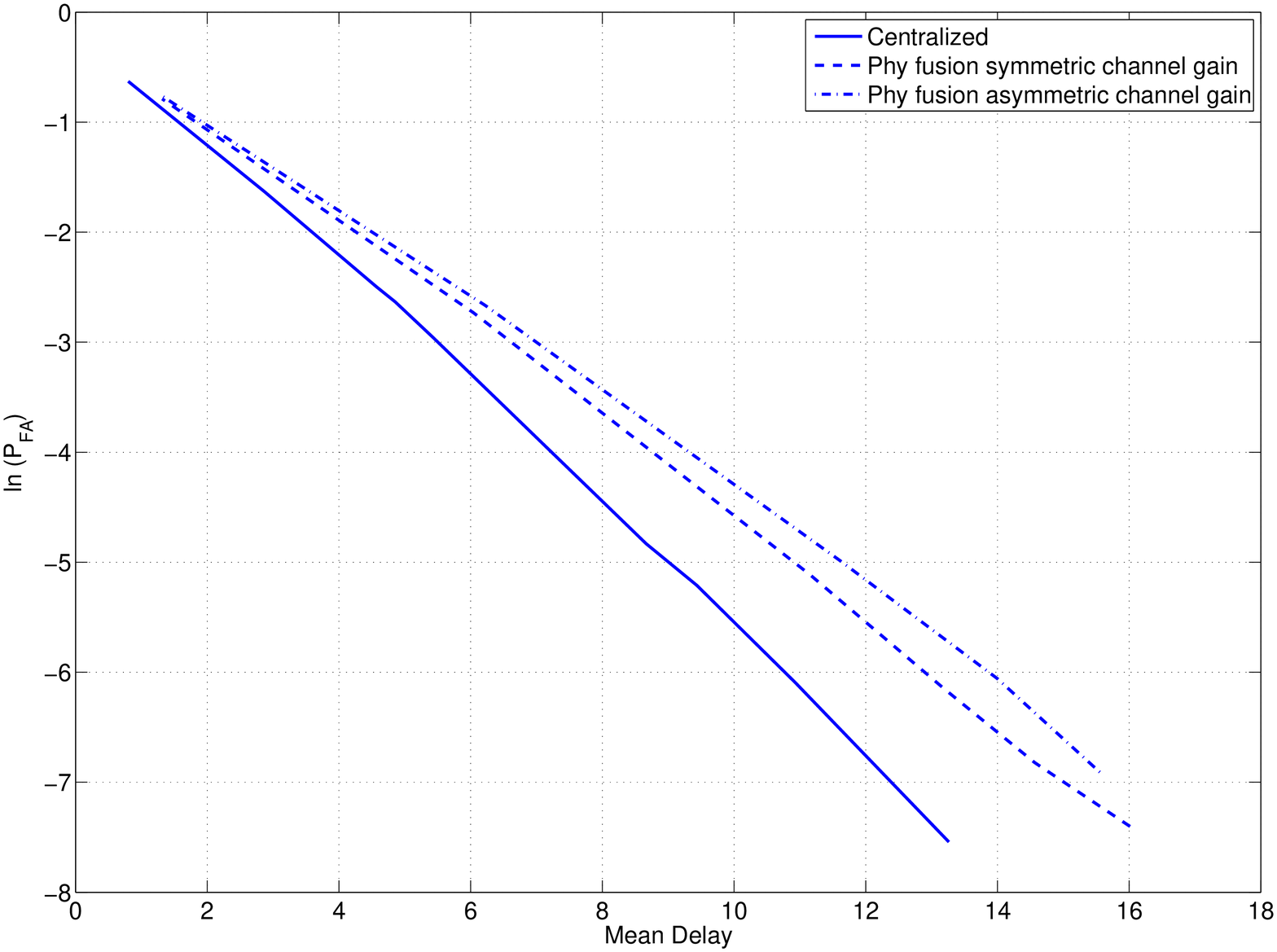}
\caption{Performance curves when 1) centralized (no channel noise)
2) symmetric channel gains 3) asymmetric channel gains with the
weaker sensor 2.5 dB lower.} \label{fig:phy_asym}
\end{figure}

\subsection{Comparison of Power- and Energy-Constrained Formulations}

For $P_{\textsf{FA}} \leq e^{-4}$, we first identify the minimum
time to detect change as a function of the energy constraint. This
yields a power constraint for the constant power formulation.
We then compare the delays incurred by the optimal algorithm under
the two formulations in Fig. \ref{fig:powerenergy}. We use the parameters in
Simulation Setup \ref{simulation} and $h_l = 1$ for all sensors. For the same
$P_{\textsf{FA}}$, the energy-constrained solution declares a change
with lesser delay than the constant power solution.

As an illustration, we plot in Fig. \ref{fig:samplepath} the variation of
$\alpha^2$, $c$, and $\mu$ with time in both the algorithms for a
representative sample path.  The change point is at $21$ samples,
shown using a dotted vertical grid line. The energy-constrained
solution is more energy efficient because it uses lower energy
($\alpha^2$) before and higher energy after the change point.
Indeed, based on the prior information, the first few samples use
negligible energy.

\begin{figure}
\centering
\includegraphics[height=2.6in,width=3.49in]{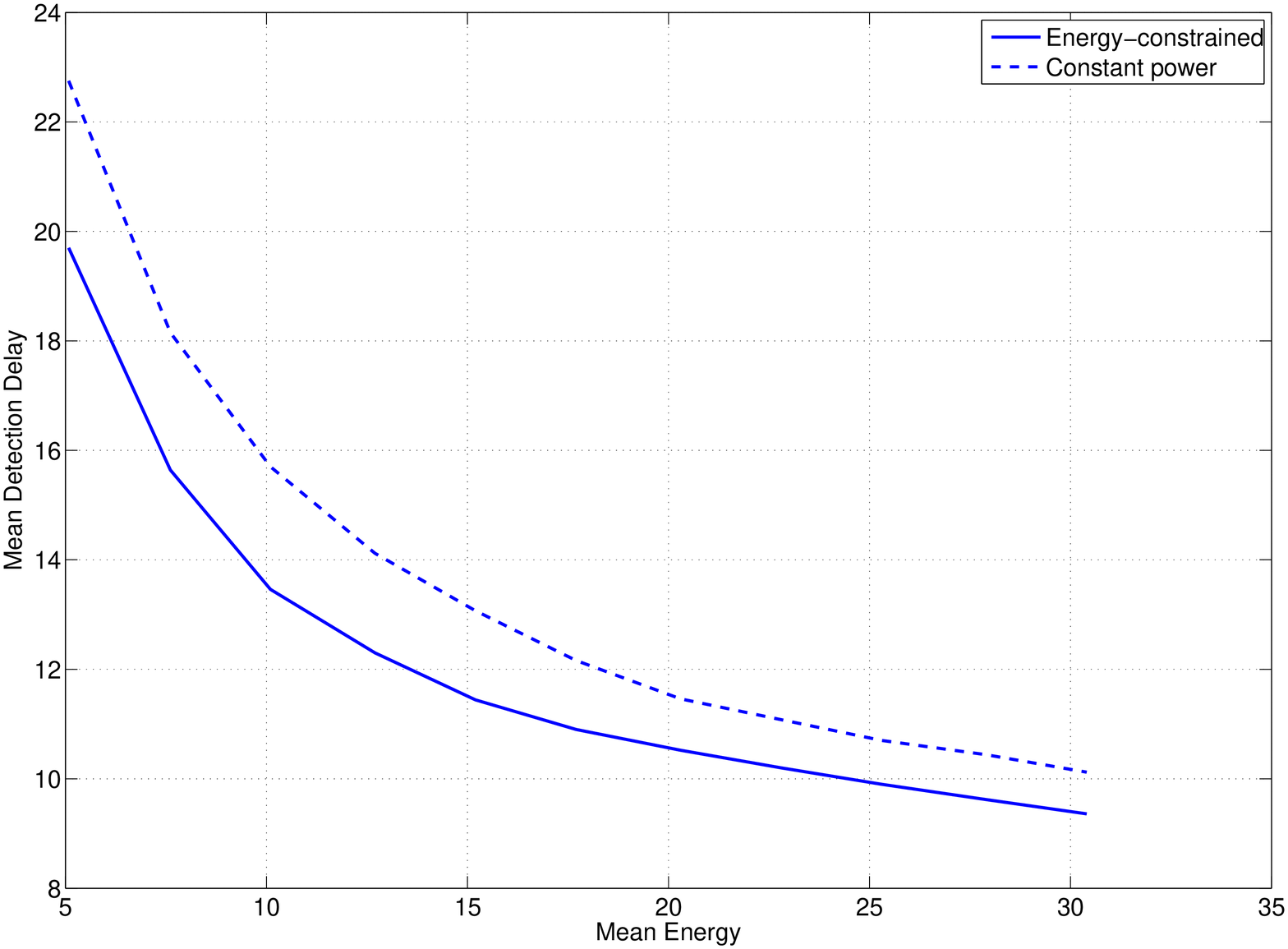}
\caption{Comparison of constant power method and energy-constrained
method.} \label{fig:powerenergy}
\end{figure}

\begin{figure}
\centering
\includegraphics[height=2.6in,width=3.49in]{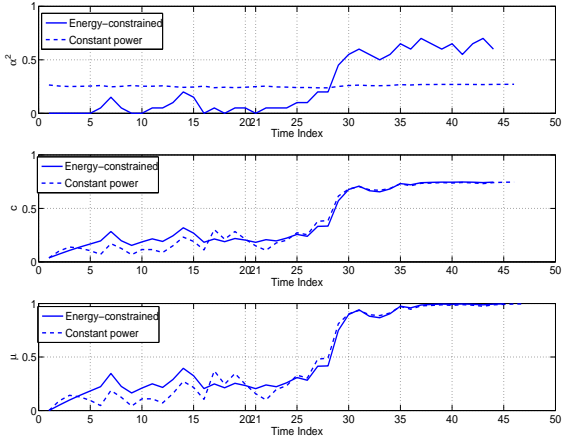}
\caption{$\alpha^2$, $c$, and $\mu$ of constant power method and
energy-constrained method for a sample path.} \label{fig:samplepath}
\end{figure}

\subsection{Channel Estimation Errors}
\label{subsec:channelerrors}

Thus far we assumed a static channel with perfect knowledge
available at both transmitter and receiver. Wireless channels,
however, change with time. Only an estimate of the channel, based on
signal processing on the pilots, beacons, or preambles, may be
available. In this section, we study the effect of imperfect channel
knowledge on the physical layer fusion algorithm.

To arrive at a model for channel errors, we consider complex channel
gains over the GMAC with noise given by $Z_{\textsf{MAC},k} \sim
\mathbb{C}\mathcal{N}(0,\sigma_{\textsf{MAC}}^2)$, a circular
symmetric complex Gaussian random variable. The observations are
real-valued, but the complex baseband equivalent signal has two
real-valued degrees of freedom per sample, leading to a bandwidth
expansion factor of two. Suppose that the sensors use transmit
beamforming\footnote{The optimality of cooperative transmit
beamforming by sensors remains an open question.}, i.e., $\alpha_l =
\frac{h_l^*} {|h_l|}\gamma_l$. Then it is sufficient to preserve
only the real part of the received signal at the fusion center, and
the problem reduces to that studied in the earlier parts of this
paper with $\sigma_{\textsf{MAC}}^2$ replaced by
$\sigma_{\textsf{MAC}}^2/2$ in Section \ref{sec:power}. The quantity
$\gamma_l$ replaces $\alpha_l$ and $|h_l|$ replaces $h_l$ in
Algorithm \ref{alg:controls}. The output of the algorithm is
$\gamma_l$.

Let $\{ h_l \}$ be a sequence of $\mathbb{C}\mathcal{N}(0,1)$ random
variables that obey a block-fading model, i.e., the channel remains
constant for $T$ uses and then changes to an independent channel
gain. If $K$ of these $T$ samples are available for channel
estimation, then the MMSE estimate of the channel is $\hat{h}_l =
(h_l+rZ)/(1+r^2)$, where $r = \sigma_{\textsf{MAC}}/\sqrt{KP_l}$,
$P_l$ is the power of sensor $l$ and $Z \sim
\mathbb{C}\mathcal{N}(0,1)$. This is estimated at both ends (using
TDD system's channel reciprocity).

Figures \ref{fig:mmse2} and \ref{fig:mmse3} show performance of the policy of Section \ref{sec:controls} with $\hat{h}$ used in place of actual $h$, across different channel SNRs. Simulation Setup \ref{simulation} parameters are used. $K=1$, i.e., only one sample pilot is used for channel estimation so that transmit beamforming is only loosely enabled. The pilot SNR equals the channel SNR in Fig. \ref{fig:mmse2} and is 8.75 dB lower in Fig. \ref{fig:mmse3}. The top-left subplot in Fig. \ref{fig:mmse2} shows that the transmit beamforming scheme with estimation errors is indeed superior to Veeravalli's scheme on a coherent GMAC. Fig. \ref{fig:mmse3} shows no benefit because the pilot SNR is not sufficient. Other subplots show graceful degradation with decreasing SNR.

\begin{figure}
\centering
\includegraphics[height=2.6in,width=3.49in]{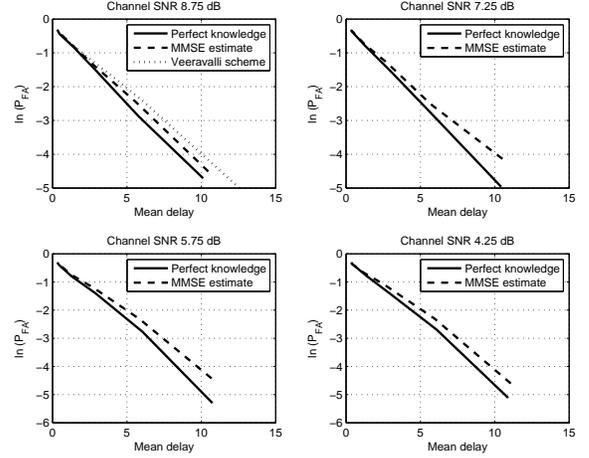}
\caption{Performance curves comparing the cases when 1) channel is perfectly known
2) MMSE estimates are used (Pilot SNR = Channel SNR).}
\label{fig:mmse2}
\end{figure}

\begin{figure}
\centering
\includegraphics[height=2.6in,width=3.49in]{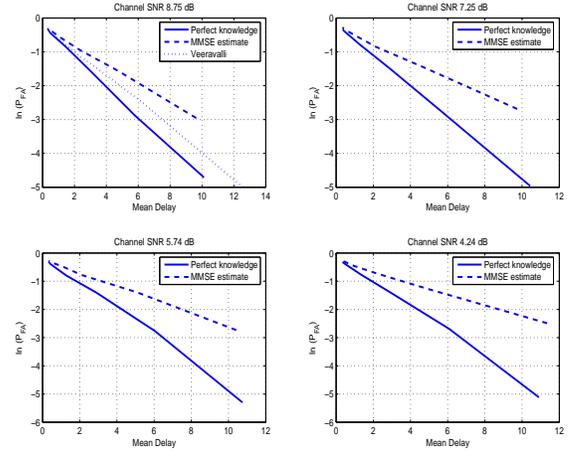}
\caption{Performance curves comparing the cases when 1) channel is perfectly known
2) MMSE estimates are used (Pilot SNR is 8.75 dB lower than Channel SNR).}
\label{fig:mmse3}
\end{figure}

\section{Summary}

We considered the use of an analog transmission strategy via an
affine transformation in order to exploit correlation in the sensor
observations. The goal was to detect a change with minimum expected
detection delay given an upper bound on the false alarm rate. We
modeled the problem as a Markov decision problem with partial
observations. We characterized the optimal control as one that
maximizes an Ali-Silvey distance between the two hypotheses before
and after the change (Appendix \ref{app:AliSilvey}). In the GMAC
setting, the optimal strategy minimizes the error variance of an
equivalent observation at the fusion center. We then gave an
explicit algorithm to identify the optimal control parameters.

We also studied a suboptimal policy that traded performance for
quantity of information fed back. We then demonstrated via
simulation the performance gain achieved by our algorithm over
another scheme that makes only a naive use of the GMAC. The latter is a
multi-access strategy optimal for independent data coupled with an
optimal distributed quantization scheme for change detection; it is
suboptimal because it does not exploit the correlation in sensor
observations. Our proposed algorithm exploits this correlation via transmit beamforming on
the GMAC. Given the control feedback in our setting, optimal
transmission strategies will change from channel to channel.
Techniques based on separation principles are therefore likely to
be suboptimal.

Distributed transmit beamforming is crucial to realize our proposed scheme.
The master-slave architecture of Mudumbai, Barriac \& Madhow
\cite{200705TWC_MudBarMad} and associated channel sensing techniques
can be used for frequency and phase synchronization. Simulations with channel estimation errors
indicate that the degradation due to lack of perfect channel
knowledge is tolerable, making this analog technique a viable
option for implementation.

We then considered a constraint on the average energy expended
instead of a power constraint. We demonstrated via simulation that
this made better use of the scarce energy resource. Extensions to
arbitrary but known distributions, in particular to the exponential
family, and to $M$-ary hypotheses can be found in \cite[Ch. 5]{200706The_Zac}.

\appendices

\section{A Characterization of Optimal Control}
\label{app:AliSilvey}

The following characterization of $A_J(\mu)$ was used in identifying
the optimal controls. The characterization refers to a
quantification of dissimilarity between probability measures called
Ali-Silvey distances (\cite{1966xxJRSS_AliSil}). Relative entropy (Kullback-Leibler divergence) is one example.
Such dissimilarity measures have a well-known
monotonicity property: data processing, whether deterministic or
random, cannot increase the dissimilarity measure between two
distributions (\cite{1966xxJRSS_AliSil}, \cite{1967xxSSMH_Csi}).
This characterization may be of interest in other
sequential detection settings.

\begin{thm}\label{thm:AliSilvey}
The minimization in (\ref{eqn:meancost}) is obtained via a
maximization of an Ali-Silvey distance between the density functions
$f_{m_1,\alpha}$ and $f_{m_0,\alpha}$.
\end{thm}

\begin{proof}
We first show that the minimization in (\ref{eqn:meancost}) can be expressed
as the maximization of an Ali-Silvey distance
$\mathbb{E}_{p_1}\left[C\left(\phi(\widehat{Y})\right)\right]$ between probability density functions
(pdf) $p_1$ and $p_2$ where
\[
  \phi(\hat{y})  =   \frac{p_2(\hat{y})}{p_1(\hat{y})}
  = \frac{f_{m_1,\alpha}(\hat{y})}{h(\hat{y},\alpha,\mu)},
\]
and $C$ is a convex function.
To see this, observe that both $p_1(.)$ and $p_2(.)$ are densities. The density $p_1$
is a mixture of pdfs under the two hypotheses while $p_2$ is the pdf
under $H_1$. Thus $g(\hat{y},\alpha,\mu)/h(\hat{y},\alpha,\mu) = \beta\phi(\hat{y})$, where $\beta = \mu+(1-\mu)p$.
From (\ref{eqn:meancost}), we have
\begin{eqnarray}
A_J(\mu)
 & = & \min_{\alpha, c} \mathbb{E}_{p_1}\left[J\left(\beta \phi(\widehat{Y})
\right)\right]
 =  \min_{\alpha, c} \mathbb{E}_{p_1}\left[G\left(\phi(\widehat{Y})
\right)\right] \nonumber \\
 & = &  1 - \max_{\alpha, c} \mathbb{E}_{p_1}\left[C\left
(\phi(\widehat{Y})\right)\right], \label{eqn:alisilvey}
\end{eqnarray}
where
$G(x) \stackrel{\triangle}{=}  J(\beta x)$ and
$C(x) \stackrel{\triangle}{=}  1-G(x)$.
$J$ is concave; so $G$ is concave, $C$ is convex, and
(\ref{eqn:alisilvey}) is obtained via a maximization of an Ali-Silvey distance between $p_1$ and $p_2$.
Now,
\begin{eqnarray}
\mathbb{E}_{p_1}\left[C\left(\phi(\widehat{Y})\right)\right] & = &
\mathbb{E}_{p_1}\left[C\left(\frac{p_2(\widehat{Y})}{p_1(\widehat{Y})}
\right)\right] \nonumber \\
& = & \mathbb{E}_{p_2}\left[\frac{p_1(\widehat{Y})}{p_2(\widehat{Y})}
C\left(\frac{p_2(\widehat{Y})}{p_1(\widehat{Y})}\right)
\right] \nonumber \\
& = & \mathbb{E}_{p_2}\left[C_1\left(\phi'
(\widehat{Y})\right)\right], \label{eqn:alisilvey1}
\end{eqnarray}
where
$
\phi'(\hat{y})  =  p_1(\hat{y}) / p_2(\hat{y}),
$
and
$
C_1(x)  =  x C \left( 1/x \right).
$
$C_1(x)$ is a convex function because $C(x)$ is convex and $x$ is nonnegative.
Now, let
$p_3(\hat{y})  =  f_{m_0,\alpha}(\hat{y})$. Since
\[
p_1(\hat{y}) / p_2(\hat{y})  = \beta+(1-\beta) p_3(\hat{y}) / p_2(\hat{y}),
\]
it is clear that
$
C_2(x) \stackrel{\triangle}{=} C_1\left(\beta+(1-\beta)x\right)
$
is a convex function. Setting $\phi''(\hat{y}) = p_3(\hat{y})/p_2
(\hat{y})$, the likelihood ratio between the original two hypotheses,
(\ref{eqn:alisilvey1}) can be written as $\mathbb{E}_{p_2}\left[C_2\left(
\phi''(\widehat{Y})\right)\right]$, an Ali-Silvey distance between
$f_{m_1,\alpha}$ and $f_{m_0,\alpha}$, and the theorem follows.
\end{proof}

\section{Proof of Lemma \ref{lemma:optimization}}
\label{app:optimization}
Here we solve Problem
\ref{prob:optialpha}. Order the indices so that
$\sigma_{\textsf{obs},1}^2h_1\alpha_{\textsf{max},1} \leq \cdots
\leq \sigma_{\textsf{obs},L}^2h_L\alpha_{\textsf{max},L}$. Let us first add a constraint
$\sum_{l=1}^Lh_l\alpha_l = a$, where without loss of generality $ a
\in \left[0, a_{\textsf{max}} \right], $ with $ a_{\textsf{max}} =
\sum_{l=1}^L h_l\alpha_{\textsf{max},l}, $ and solve the convex
optimization problem:

\vspace*{.1in}

\begin{problem}
\label{prob:convex}
Minimize
$
  \sum_{l=1}^L\sigma_{\textsf{obs},l}^2h_l^2\alpha_l^2
$
subject to
$
  \alpha_l \in \left[0, \alpha_{\textsf{max},l} \right],
$
$
\sum_{l=1}^Lh_l\alpha_l = a \in \left[0, a_{\textsf{max}}\right].
$
\end{problem}

\vspace*{.1in}

This problem is a special case of a separable convex
optimization problem studied in Padakandla and Sundaresan
\cite{200707SIOPTarXiv_PadSun}. Execution of \cite[Algorithm
1]{200707SIOPTarXiv_PadSun} yields the following solution. Break
$[0,a_\textsf{max}]$ into $L$ intervals $[a_k,a_{k+1}], k =
0,1,\ldots,L-1$, where $a_0 = 0$ and
\[
a_k  =
\left(\sum_{l=1}^kh_l\alpha_{\textsf{max},l}+\sigma_{\textsf{obs},k}^2h_k\alpha_{\textsf{max},k}
\sum_{l=k+1}^L \sigma_{\textsf{obs},l}^{-2}\right).
\]
The ordering of $\sigma_{\textsf{obs},l}^2h_l\alpha_{\textsf{max},l}$
implies $a_{m}$$ \leq a_{m+1}$ so that each interval is nonempty.
With $k$ such that $a \in [a_k, a_{k+1}]$, the optimal solution is:
\begin{eqnarray}
\label{eqn:alphalower}\alpha_l & = & \alpha_{\textsf{max},l},\quad l
= 1,
\ldots,k,\\
\label{eqn:alphaupper}\alpha_l & = & \frac{1}{\sigma_{\textsf{obs},l}^2h_l} \cdot \frac{a-\sum_{m=1}
^kh_m\alpha_{\textsf{max},m}}{\sum_{m=k+1}^L
\sigma_{\textsf{obs},m}^{-2}},\quad l > k.
\end{eqnarray}
The corresponding minimum value of Problem \ref{prob:convex} for a given $a$, denoted by $V(a)$, is given by
\[
V(a) =
\sum_{l=1}^k\sigma_{\textsf{obs},l}^2h_l^2\alpha_{\textsf{max},l}^2+\frac{\left(a-
\sum_{l=1}^kh_l\alpha_
{\textsf{max},l}\right)^2}{\sum_{l=k+1}^L\sigma_{\textsf{obs},l}^{-2}}.
\]
We next look for an optimal $a$ by solving

\vspace*{.1in}

\begin{problem}
Minimize $ f(a) = \frac{V(a) + \sigma_{\textsf{MAC}}^2}{a^2} $
subject to $a \in [0, a_{\max}]$.
\end{problem}

\vspace*{.1in}

While this is not yet a convex optimization, the transformation $b=1/a$ casts it into one. Define
\begin{eqnarray*}
\lefteqn{ g(b) =  f\left(\frac{1}{a}\right) } \\
&& =
b^2\left[\sum_{l=1}^k\sigma_{\textsf{obs},l}^2h_l^2\alpha_{\textsf{max},l}^2\!+\!
\sigma_\textsf{MAC}^2\!+\!\frac{\left(\sum_{l=1}^kh_l\alpha_{\textsf{max},l}\right)^2}
{\sum_{l=k+1}^L\sigma_{\textsf{obs},l}^{-2}}\right] \\
& & ~
- ~ 2b \cdot \frac{\sum_{l=1}^kh_l\alpha_{\textsf{max},l}}{\sum_{l=k+1}^L \sigma_{\textsf{obs},l}^{-2}}
+\frac{1}{\sum_{l=k+1}^L \sigma_{\textsf{obs},l}^{-2}},
\end{eqnarray*}
for $b \in [1/a_{\max}, \infty)$,
where $k$ depends on $b$ through the index of the interval in which $a = 1/b$ lies. The following observations on $g$ are easy to verify:
\begin{itemize}
\item $g(b)$ is a convex parabola on each $[1/a_{k+1}, 1/a_k], k = L-1, \cdots, 1$, and on $[1/a_1, \infty)$;
\item $g(b)$ is continuous in $[1/a_{\max},\infty)$. This needs checking only at interval boundaries $1/a_k$;
\item $g(b)$ is continuously differentiable in $(1/a_{\max},\infty)$ with left continuous derivative at $1/a_{\max}$;
\item $\lim_{b \rightarrow \infty} g'(b) = + \infty$, so that the derivative eventually becomes positive for large $b$.
\end{itemize}
Since $g$ is convex and continuously differentiable, if we can find a $b^*$ such that $g'(b^*) = 0$ and $b^* \in [1/a_{k+1}, 1/a_k]$ (or $[1/a_1, \infty)$) where $k$ corresponds to $a^* = 1/b^*$, then $b^*$ is a point of global minimum. This holds if the minimum point for a parabola defined in $[1/a_{k+1}, 1/a_k]$ (or $[1/a_1, \infty)$), which is easily verified to be
\begin{eqnarray*}
\lefteqn{ a^* = 1/b^* =
\sum_{l=1}^kh_l\alpha_{\textsf{max},l} } \\
&& + ~\frac{\sum_{l=k+1}^L \sigma_{\textsf{obs},l}^{-2}}{\sum_{l=1}^kh_l\alpha_{\textsf{max},l}} \cdot
\left( \sum_{l=1}^k(\sigma_{\textsf{obs},l}h_l\alpha_{\textsf{max},l})^2
 +\sigma_{\textsf{MAC}}^2 \right),
\end{eqnarray*}
also belongs to that interval. This leads to the condition (\ref{eqn:condition}). If no such point occurs, $g'(b) \neq 0$ in $[1/a_{\max}, \infty)$, and since $g'$ is eventually positive, it must be positive in the entire interval. In this latter case $g$ is an increasing function on $[1/a_{\max}, \infty)$ and the minimum is attained at $b^* = 1/a_{\max}$ or $a^* = a_{\max}$ or $k = L$. Substitution of $a^*$ in (\ref{eqn:alphalower}) and (\ref{eqn:alphaupper}) completes the proof.

\bibliography{journalPaper_revised_v6.bbl}

\begin{thebibliography}{10}
\providecommand{\url}[1]{#1}
\csname url@rmstyle\endcsname
\providecommand{\newblock}{\relax}
\providecommand{\bibinfo}[2]{#2}
\providecommand\BIBentrySTDinterwordspacing{\spaceskip=0pt\relax}
\providecommand\BIBentryALTinterwordstretchfactor{4}
\providecommand\BIBentryALTinterwordspacing{\spaceskip=\fontdimen2\font plus
\BIBentryALTinterwordstretchfactor\fontdimen3\font minus
  \fontdimen4\font\relax}
\providecommand\BIBforeignlanguage[2]{{%
\expandafter\ifx\csname l@#1\endcsname\relax
\typeout{** WARNING: IEEEtran.bst: No hyphenation pattern has been}%
\typeout{** loaded for the language `#1'. Using the pattern for}%
\typeout{** the default language instead.}%
\else
\language=\csname l@#1\endcsname
\fi
#2}}

\bibitem{200705TWC_MudBarMad}
R.~Mudumbai, G.~Barriac, and U.~Madhow, ``On the feasibility of distributed
  beamforming in wireless networks,'' \emph{{IEEE} Trans. Wireless Commun.},
  vol.~6, no.~5, pp. 1754--1763, May 2007.

\bibitem{1966xxJRSS_AliSil}
S.~M. Ali and S.~D. Silvey, ``A general class of coefficients of divergence of
  one distribution from another,'' \emph{J. Royal Statist. Soc., Ser. B},
  vol.~28, no.~1, pp. 131--142, 1966.

\bibitem{195406BIO_Pag}
E.~S. Page, ``Continuous inspection schemes,'' \emph{Biometrika}, vol.~41, no.
  1/2, pp. 100--115, Jun. 1954.

\bibitem{197112AMS_Lor}
G.~Lorden, ``Procedures for reacting to a change in distribution,'' \emph{Ann.
  Math. Statist.}, vol.~42, pp. 1897--1908, Dec. 1971.

\bibitem{Shi-OSR78}
A.~N. Shiryayev, \emph{Optimal Stopping Rules}.\hskip 1em plus 0.5em minus
  0.4em\relax New York: Springer-Verlag, 1978.

\bibitem{200105TIT_Vee}
V.~V. Veeravalli, ``Decentralized quickest change detection,'' \emph{{IEEE}
  Trans. Inform. Theory}, vol.~47, no.~4, pp. 1657--1665, May 2001.

\bibitem{1990xxDD_Tsi}
J.~N. Tsitsiklis, \emph{Advances in Statistical Signal Processing}.\hskip 1em
  plus 0.5em minus 0.4em\relax Greenwich, CT: JAI Press, 1993, vol.~2, ch.
  Decentralized Detection, pp. 297--344.

\bibitem{199303TIT_VeeBasPoo}
V.~V. Veeravalli, T.~Ba\c{s}ar, and H.~V. Poor, ``Decentralized sequential
  detection with a fusion center performing the sequential test,'' \emph{{IEEE}
  Trans. Inform. Theory}, vol.~39, no.~2, pp. 433--442, Mar. 1993.

\bibitem{200606The_PraKum}
V.~K. Prasanthi, ``Towards a cross layer design of sequential change detection
  over ad hoc wireless sensor networks,'' Master's thesis, Indian Institute of
  Science, Bangalore, India, June 2006.

\bibitem{200609SECON_PraKum}
V.~K. Prasanthi and A.~Kumar, ``Optimizing delay in sequential change detection
  on ad hoc wireless sensor networks,'' in \emph{Proc. IEEE SECON}, Reston, VA,
  Sep. 2006.

\bibitem{Ber-RDT71}
T.~Berger, \emph{Rate Distortion Theory: A Mathematical Basis for Data
  Compression}.\hskip 1em plus 0.5em minus 0.4em\relax Englewood Cliffs:
  Prentice-Hall, 1971.

\bibitem{200607ISIT_LapTin}
A.~Lapidoth and S.~Tinguely, ``Sending a bi-variate {G}aussian source over a
  {G}aussian {MAC},'' in \emph{Proc. 2006 IEEE International Symposium on
  Information theory}, Seattle, Washington, Jul. 2006.

\bibitem{200602TSP_MerTon}
G.~Mergen and L.Tong, ``Type based estimation over multiaccess channels,''
  \emph{{IEEE} Trans. Signal Processing}, vol.~54, no.~2, pp. 613--626, Feb.
  2006.

\bibitem{200703TSP_MerNawTon}
G.~Mergen, V.~Naware, and L.~Tong, ``Asymptotic detection performance of
  type-based multiple access over multiaccess fading channels,'' \emph{{IEEE}
  Trans. Signal Processing}, vol.~55, no.~3, pp. 1081--1092, Mar. 2007.

\bibitem{200705TSP_LiuSay}
K.~Liu and A.~M. Sayeed, ``Type-based decentralized detection in wireless
  sensor networks,'' \emph{{IEEE} Trans. Signal Processing}, vol.~55, no.~5,
  pp. 1899--1910, May 2007.

\bibitem{200707TSP_LiuGamSay}
K.~Liu, H.~E. Gamal, and A.~Sayeed, ``Decentralized inference over
  multiple-access channels,'' \emph{{IEEE} Trans. Signal Processing}, vol.~55,
  no.~7, pp. 3445--3455, Jul. 2007.

\bibitem{200310Allerton_GasVet}
M.~Gastpar and M.~Vetterli, ``Scaling laws for homogeneous sensor networks,''
  in \emph{Proceedings of the 41st Ann. Allerton Conf. Commun., Contr.,
  Comput.}, Monticello, IL, Oct. 2003.

\bibitem{200305TIT_GasRimVet}
M.~Gastpar, B.~Rimoldi, and M.~Vetterli, ``To code, or not to code: Lossy
  source-channel communication revisited,'' \emph{{IEEE} Trans. Inform.
  Theory}, vol.~49, pp. 1147--1158, May 2003.

\bibitem{200504JSAC_GasVet}
M.~Gastpar and M.~Vetterli, ``Power, spatio-temporal bandwidth, and distortion
  in large sensor networks,'' \emph{{IEEE} J. Select. Areas Commun.}, vol.~23,
  pp. 745--754, Apr. 2005.

\bibitem{199802AFLS_ErtPot}
E.~Ertin and L.~C. Potter, ``Decentralized detection with generalized costs,''
  in \emph{Proc. Second ARL Feder. Lab. Symp.}, College Park, MD, Feb. 1998.

\bibitem{200706The_Zac}
L.~Zacharias, ``Decentralized sequential detection using physical layer
  fusion,'' Master's thesis, Indian Institute of Science, Bangalore, India,
  June 2007.

\bibitem{1978SOC_BerShr}
D.~P. Bertsekas and S.~E. Shreve, \emph{Stochastic Optimal Control: The
  Discrete Time Case}.\hskip 1em plus 0.5em minus 0.4em\relax NY: Academic
  Press, 1978.

\bibitem{Ber-DPOC95}
D.~P. Bertsekas, \emph{Dynamic Programming and Optimal Control, Volume
  1}.\hskip 1em plus 0.5em minus 0.4em\relax Massachusetts: Athena Scientific,
  1995.

\bibitem{200706ISIT_ZacSun}
L.~Zacharias and R.~Sundaresan, ``Decentralized sequential change detection
  using physical layer fusion,'' in \emph{Proc. 2007 IEEE Int. Symp. on Inform.
  Theory}, Nice, France, Jun. 2007.

\bibitem{1967xxSSMH_Csi}
I.~Csisz\'{a}r, ``Information-type measures of difference of probability
  distributions and indirect observations,'' \emph{Studia Sci. Math. Hungar.},
  vol.~2, pp. 299--318, 1967.

\bibitem{200707SIOPTarXiv_PadSun}
A.~Padakandla and R.~Sundaresan, ``Separable convex optimization problems with
  linear ascending constraints,'' \emph{Submitted to SIAM J. on Opt. {\tt
  http://arxiv.org/abs/0707.2265}}, Jul. 2007.

\end{thebibliography}

\end{document}